\documentclass[journal,letterpaper,twocolumn,twoside,nofonttune]{IEEEtran}
\usepackage{etex}
\normalsize
\usepackage[T1]{fontenc}
\usepackage{amsmath,amssymb,amsfonts}
\usepackage{mathrsfs}
\usepackage{mathabx}
\usepackage{amsbsy}
\usepackage{graphicx}
\usepackage{graphics}
 \usepackage{epstopdf}
\usepackage{algorithm}
\usepackage{algorithmic}
\usepackage{subfigure}
\usepackage{cite}
\usepackage{multirow}

\usepackage{enumerate}
\usepackage{caption}

\usepackage{setspace}
\usepackage{mathtools}
\usepackage{lipsum}
\usepackage[algo2e]{algorithm2e}
\usepackage{diagbox}
\usepackage[thinc]{esdiff}
\usepackage{tikz,pgfplots}
\usetikzlibrary{spy}
\usepackage{ctable}

\title{\Huge$\,$\\[-2.75ex]
{ List-Decodable Coded Computing:\\ Breaking the Adversarial Toleration Barrier  }\\[0.50ex]}

\author{\large%
	Mahdi Soleymani, Ramy E. Ali, Hessam Mahdavifar, and A. Salman Avestimehr
	\vspace{-.25in}
	\thanks{This material is based upon work supported by Defense Advanced Research Projects Agency (DARPA) under Contract No. HR001117C0053, ARO award W911NF1810400, NSF grants CCF-1703575, CCF-1763673, CCF--1763348, CCF--1909771 and CCF--1941633, MLWINS-2002874, ONR Award No. N00014-16-1-2189, and a gift from Intel/Avast/Borsetta via the PrivateAI institute. The views, opinions, and/or findings expressed are those of the author(s) and should not be interpreted as representing the official views or policies of the Department of Defense or the U.S. Government. 
	
	This article was presented in part at
the proceedings of the 2021 IEEE International Symposium on Information Theory. The associate editor coordinating the review of this article
and approving it for publication was P. Grover. (Corresponding author:
Mahdi Soleymani.)
	
	M.\ Soleymani and H.\ Mahdavifar are with the Department of Electrical Engineering and Computer Science, University of Michigan, Ann Arbor, MI 48104 (email: mahdy@umich.edu and hessam@umich.edu).}
	\thanks{R.E.\ Ali and A.\ Salman Avestimehr are with the Department of Electrical
		Engineering, University of Southern California, Los Angeles, CA 90089 USA
		(e-mail: reali@usc.edu and avestimehr@ee.usc.edu).}
}

\newtheorem{theorem}{{Theorem}}
\newtheorem{lemma}{Lemma}

\newtheorem{corollary}{Corollary}

\newtheorem{definition}{{Definition}}

\newtheorem{remark}{{\textbf{Remark}}}
\AfterEndEnvironment{remark}{\noindent\ignorespaces}

\DeclareMathAlphabet{\mathbfsl}{OT1}{ppl}{b}{it} 

\newcommand{\bA}{\mathbfsl{A}}

\newcommand{\bI}{\mathbfsl{I}}

\newcommand{\bM}{\mathbfsl{M}} 
\newcommand{\bN}{\mathbfsl{N}}

\newcommand{\bV}{\mathbfsl{V}}
 
\newcommand{\bX}{\mathbfsl{X}}
\newcommand{\bY}{\mathbfsl{Y}} 
\newcommand{\bZ}{\mathbfsl{Z}}

\newcommand{\bc}{\mathbfsl{c}}

\newcommand{\bv}{\mathbfsl{v}}
 
\newcommand{\bx}{\mathbfsl{x}}
\newcommand{\by}{\mathbfsl{y}} 
\newcommand{\bz}{\mathbfsl{z}}

 \newcommand{\bff}{\mathbfsl{f}} 
 
\newcommand{\ceil}[1]{\left\lceil #1 \right\rceil}
\newcommand{\floor}[1]{\left\lfloor #1 \right\rfloor}

\newcommand*{\rom}[1]{\expandafter\romannumeral #1}

\makeatletter
\newcommand{\AlignFootnote}[1]{%
  \ifmeasuring@
  \else
    \iffirstchoice@
      \footnote{#1}%
    \fi
  \fi}
\makeatother

\usepackage[utf8]{inputenc}

\newcommand{\be}[1]{\begin{equation}\label{#1}}
\newcommand{\ee}{\end{equation}}

\renewcommand{\leq}{\leqslant}
 
\renewcommand{\geq}{\geqslant}

\renewcommand{\Bbb}{\mathbb}

\newcommand{\R}{{\Bbb R}}

\newcommand{\Tref}[1]{Theo\-rem\,\ref{#1}}

\newcommand{\Lref}[1]{Lem\-ma\,\ref{#1}}
\newcommand{\Cref}[1]{Co\-ro\-lla\-ry\,\ref{#1}}

\newcommand{\Fq}{{{\Bbb F}}_{\!q}}

\newcommand{\deff}{\mbox{$\stackrel{\rm def}{=}$}}

\newcommand{\Span}[1]{{\left\langle {#1} \right\rangle}}

\DeclareMathOperator*{\argmax}{arg\,max}

\begin{document}

\maketitle

\begin{abstract}
We consider the problem of coded computing, where a computational task is performed in a distributed fashion in the presence of adversarial workers. We propose techniques to break the adversarial toleration threshold barrier previously known in coded computing. More specifically, we leverage list-decoding techniques for folded Reed-Solomon codes and propose novel algorithms to recover the correct codeword using side information. In the coded computing setting, we show how the master node can perform certain carefully designed extra computations to obtain the side information. The workload of computing this side information is negligible compared to the computations done by each worker.  This side information is then utilized to prune the output of the list decoder and uniquely recover the true outcome. We further propose folded Lagrange coded computing (FLCC) to incorporate the developed techniques into a specific coded computing setting. Our results show that FLCC outperforms LCC by breaking the barrier on the number of adversaries that can be tolerated. In particular, the corresponding threshold in FLCC is improved by a factor of two asymptotically compared to that of LCC. 
\end{abstract}

\section{Introduction}

Recently, ideas from the coding theory literature have been widely leveraged in large-scale distributed computing and learning problems to alleviate major performance bottlenecks including latency in computations, communication overheads, and stragglers \cite{lee2017speeding,li2017coding,yu2019lagrange,li2020coded}. This has led to the emergence of the \emph{coded computing} paradigm by combining coding theory and distributed computing, also addressing critical issues such as security and privacy in distributed settings. More specifically, there has been an increasing interest in recent years toward adopting coded computing techniques in computationally-demanding machine learning tasks that give rise to several privacy and security issues \cite{so2019codedprivateml, so2020turbo, so2020scalable, soleymani2020privacy, jahani2020berrut, prakash2020coded}. In such settings, the underlying dataset must remain private from the cloud and the contributing computational workers, as it may contain highly sensitive information such as biometric data of patients in a hospital \cite{raghupathi2014big} or customers' data of a company \cite{mcafee2012big}. Moreover, the outcome of a distributed learning scheme, e.g., model parameters trained on a dataset, must be \emph{secured} against Byzantine (malicious) adversaries that attack the cloud or are present as adversarial workers aiming at altering the outcome either for their benefits or to deceive the other users.  

A well-established architecture often considered in coded computing consists of a \emph{master} node and a set of workers having communication links with the master node. The goal for the master is to perform a certain computational job, e.g., training a model on its dataset, with the help of the workers. To this end, the master disperses its dataset among the workers that operate in parallel and return their results to the master to recover the outcome of the computational job efficiently. Then a problem of significant interest is the following: what fraction of adversarial workers can be tolerated, i.e., the master is still able to recover the true outcome even though the adversaries have returned corrupted results, in such coded computing schemes? To answer this question, the well-known classical results on the error correction capability of linear codes and, in particular, maximum distance separable (MDS) codes such as Reed-Solomon (RS) codes are leveraged, see, e.g., \cite{yu2019lagrange}. The tightness of such results on adversarial toleration is based on certain assumptions on the underlying code and the corresponding decoder employed by the master node. For instance, it is implicitly assumed that the master node performs the decoding only given the returned results by the workers and does not perform any extra computations to gain \textit{side information} about the computation outcome. Also, the decoder employed by the master is assumed to be the classical decoder that recovers errors up to half the minimum distance bound. However, the list-decoding paradigm offers the potential to decode errors beyond this bound \cite{elias1957list}.  In fact, there is a long history on list-decoding algorithms for RS codes with the end result of achieving the information-theoretic Singleton bound $1-R$, where $R$ is the code rate, on the decoding radius for a variant of RS codes, called \textit{folded} RS (FRS) codes \cite{sudan1997decoding, guruswami1998improved, parvaresh2005, guruswami2008explicit}. This improves upon the half the minimum distance bound, expressed as $(1-R)/2$ for MDS codes, by a factor of $2$ closing the gap with the Singleton bound \cite{guruswami2008explicit}. 

We consider the following fundamental question in this paper: is it possible to break the adversarial toleration threshold barrier established in the coded computing literature? We show that the answer to this question is yes. To this end, we leverage the advances in the list-decoding literature as well as the particular coded computing setting that naturally allows the master node to have access to side information and uniquely determine the computation outcome.

\subsection{Our contributions}
 In this paper, we show how to adapt the \textit{folding} technique in algebraic coding to the realm of coded computing, where the underlying computational job is a polynomial evaluation over the dataset. Then it is shown how the master node can employ an off-the-shelf FRS list-decoding algorithm, e.g., the linear-algebraic algorithm proposed in \cite{guruswami2013linear}, to the results returned from the workers. This results in a low-dimensional linear subspace that contains the true outcome of the computation assuming a certain bound on the number of adversaries. In order to uniquely recover the true outcome, we propose two schemes which involve the master node performing certain carefully designed extra computations to obtain side information about the outcome. This side information is then utilized to prune the subspace of possible outcomes, i.e., the output of the FRS list-decoding algorithm, to uniquely recover the true outcome. In both schemes, the cost of computing the side information is negligible compared to the computation load of each worker. 
 Specifically, our contributions are as follows.
 \begin{enumerate}[1) ]
     \item We propose a deterministic pruning algorithm, in which the master node waits until the results are returned by the worker nodes and the FRS list-decoding algorithm is applied to the returned results. Then the master node carefully selects a certain small subset of evaluation points and computes the polynomial evaluation over these points to obtain the side information.  It is shown that this can be done in such a way that the true outcome is uniquely recovered from the output list. 

     \item We also propose a probabilistic pruning algorithm in which the side information is obtained by computing the polynomial evaluation over a randomly selected set of evaluation points. This can be done in parallel to the tasks being performed by the workers resulting in a lower latency compared to the first approach. Then it is shown that the true outcome can be uniquely recovered with a \emph{high} probability. Moreover, if the unique recovery is not possible in this case, the master node can identify it as a decoding failure. We show results outperforming the state-of-the-art schemes in terms of the lower bound on the probability of successful decoding and the amount of side information needed for unique decoding.     
     \item To illustrate how our proposed protocols break the adversarial toleration thresholds in a certain class of coded computing schemes, we consider the Lagrange coded computing (LCC) scheme. We introduce a \emph{folded} version of LCC, referred to as FLCC. Similar to other mainstream coded computing schemes, the master node in LCC attempts to decode the computation outcome  merely based on the  computations of the workers. By relaxing this restriction, the master node in FLCC invokes our proposed pruning algorithms together with list-decoding FRS codes and uniquely recovers the outcome. The performance of FLCC with both the deterministic and the probabilistic pruning algorithms is characterized and compared to that of LCC. Our results indicate that the cost of overcoming a Byzantine worker in FLCC can be reduced to be almost the same as that of a straggler worker, in the characterization of recovery thresholds, as opposed to LCC in which Byzantine adversaries cost twice as stragglers.  

 \end{enumerate}

\subsection{Related work} 

The problem of list-decoding with side information was initially considered in \cite{guruswami2003list} for binary codes in a communication setup. In this setting, a clean noise-free channel is assumed over which a \emph{small} amount of side information, compared to the length of the message, is provided to the receiver. This side information consists of a random hash function along with its value over the message.  Another variant of RS codes, known as \emph{derivative} codes that also achieve the \textit{optimal} performance have been studied in \cite{kopparty2015list,guruswami2013linear}.   In 
\cite{guruswami2013linear}, a linear-algebraic list-decoding approach along with pruning algorithms with side information specific to derivative codes were proposed to uniquely recover the codeword either deterministically or 
with high probability. This approach also can list-decode FRS codes with side information to recover the codeword with high probability. However, this approach was not extended to deterministically recover the codeword. The reason for this is that to deterministically recover the codeword using this approach, the side information may not be decided before starting to decode. While this is inconvenient in communication setups, it is possible in coded computing as the master node plays the role of both the encoder and the decoder.       
 List-decoding of FRS codes has been also incorporated in the context of secret sharing to enhance the security \cite{cramer2015linear,safavi2015model,cheraghchi2019nearly}. However, these works do not consider computations over data and are only concerned with recovering the data from the secret shares. 

Coded computing has recently gained much interest due to its promise to overcome several issues raised in large-scale distributed computing and machine learning.
It has been utilized for straggler mitigation in various distributed computing tasks \cite{li2016unified,yu2020straggler,reisizadeh2019coded,aliasgari2018coded, jamali2019coded,yang2021edge}. Several schemes for distributed matrix-matrix multiplication, which is one of the main building blocks for various machine learning algorithms, have been also proposed in the literature \cite{yu2020entangled,aliasgari2020private,d2020gasp,bitar2019private,nodehi2019secure}. Moreover, certain protocols have been introduced for computations over real-valued data \cite{fahim2019numerically,ramamoorthy2019numerically,das2019distributed,charalambides2020numerically}. Also, recently, analog coded computing protocols have been introduced to enable privacy for large-scale distributed machine learning in the analog domain \cite{soleymani2020privacy, soleymani2020analog}. Improving the adversarial toleration threshold of LCC has been also considered recently in \cite{yang2021coded} for Boolean computations and sparse polynomials and in \cite{Tang2021} for matrix-matrix multiplication. In addition, this problem was also studied for the probabilistic noise model in \cite{dutta2018unified,subramaniam2019collaborative}. Our work, however, considers any polynomial-based computations and the worst-case adversarial model. Moreover, none of these prior works incorporated list-decoding ideas into the coded computing protocols.

The rest of this paper is organized as follows. In Section\,\ref{sec:background}, some background on list-decoding of RS codes and their variants is provided. The system model considered in this paper is discussed in Section\,\ref{Section: System Model} and FLCC is proposed. Our results on list-decoding of FRS codes with side information are shown in Section\,\ref{sec:listdecoding-results}. In Section\,\ref{Section: Folded Lagrange Coded Computing} it is shown how our techniques applied to FLCC improve upon the security of LCC against Byzantine adversaries. Finally, the paper is concluded in Section\,\ref{sec:Conclusion}.

\section{Background}\label{sec:background}
In this section, we first provide a brief background on list-decoding.   Then, we briefly overview Lagrange coded computing (LCC) that will be used later in Section\,\ref{Section: Folded Lagrange Coded Computing} to showcase how  the list decoding ideas are leveraged to improve the adversarial toleration threshold in a well-established  coded computing scheme.  
\subsection{List decoding FRS codes} \label{Pre_FRS}
We begin by introducing the notations that are used throughout this paper. For a positive integer $i$, the set $\{1, 2, \cdots, i \}$ is denoted by $[i]$. The number of positions at which two strings of length $n$, $\by$ and $\by'$, differ is denoted by $\Delta(\by, \by') \deff |\{i: y_i \neq y'_i \}|$ and their relative distance is denoted by $\delta(\by, \by') \deff \Delta(\by, \by')/n$. A code $C: \Sigma^k \mapsto \Sigma^n$ of length $n$ over alphabet $\Sigma$, where $|\Sigma|=q$, is denoted by $[n, k]_q$.  A finite field of size $q$ is denoted by $\Fq$ and the set of all non-zero elements of $\Fq$ is represented by $\Fq^*$.

An $[n, k]_q$ MDS code such as a Reed-Solomon code can always correct up to $\rho_{\mathrm U}(R) = (1-R)/2$ normalized number of errors, also referred to as decoding radius, where $R\,\deff\,k/n$ is the rate of the code and normalization is done by dividing the number of errors by $n$. In order to correct errors beyond this bound, list-decoding \cite{elias1957list, wozencraft1958list, goldreich1989hard, sudan1997decoding} relaxes the unique decoding requirement and allows the decoder to output a list of codewords. Specifically, given $0 \leq \rho \leq 1$, an $[n, k]_q$ code $ C \subseteq \Sigma^n$ is said to be $(\rho, L)$-list decodable if for every $\by \in \Sigma^n$, the set $\mathcal L\,\deff\,\{\bc \in \ C | \delta(\by, \bc) \leq \rho n\}$ has at most $L$ elements. Based on this relaxation, Guruswami and Sudan \cite{guruswami1998improved}  showed that Reed-Solomon codes can be list-decoded up to the decoding radius of $\rho_{\mathrm{GS}}(R)=1-\sqrt{R}$. However, it is well-known that there exist codes that can be list-decoded up to a decoding radius of $1-R-\epsilon$ with a list size of at most $O(1/\epsilon)$\cite{elias1991error, guruswami2002combinatorial}. Parvaresh and Vardy further improved Guruswami-Sudan decoding radius by introducing a variant of Reed-Solomon codes, also referred to as Parvaresh-Vardy codes \cite{parvaresh2005}, followed by Guruswami and Rudra \cite{guruswami2008explicit} who showed that such variations can be more efficiently realized by folding the Reed-Solomon codes, thereby improving the decoding radius to approach the ultimate Singleton bound $1-R$. Next, the definition of folded Reed-Solomon codes is provided.
\begin{definition}\label{FRS-deff}
($m$-Folded Reed-Solomon Code \cite{guruswami2008explicit}) Let $\gamma$ be a primitive element of $\mathbb F_q$, $n \leq q-1$ be a multiple of $m$, and $k$ with $1\leq k<n$ be the degree parameter. The folded Reed-Solomon (FRS) code $\mathrm{FRS}_q^{(m)}[n,k]$ is a code over alphabet $\mathbb F_q^m$ that encodes a polynomial $f \in \mathbb F_q[X]$ of degree at most $k-1$ as 
\begin{align*}
\small{
\left( \left [ \begin{array}{c}
f(1)\\
f(\gamma)\\
\vdots\\
f(\gamma^{(m-1)})
\end{array} \right], 
\left [ \begin{array}{c}
	f(\gamma^m)\\
	f(\gamma^{m+1})\\
	\vdots\\
	f(\gamma^{(2m-1)})
\end{array} \right], 
\hdots,
\left [ \begin{array}{c}
	f(\gamma^{n-m})\\
	f(\gamma^{n-m+1})\\
	\vdots\\
	f(\gamma^{(n-1)})
\end{array} \right]
\right ),}
\end{align*}
where the block length of  $\mathrm{FRS}_q^{(m)}[n,k]$ is $N=n/m$, and its rate is $R=k/n$.
\end{definition}
\noindent Guruswami and Rudra \cite{guruswami2008explicit} showed that FRS codes can be efficiently list-decoded up to the decoding radius $\rho_{GR}(R)=1-R-\epsilon$ with a list size of $L=n^{O(1/ \epsilon)}$. Later, it was shown in \cite{guruswami2013linear} that this can be done using an alternative linear-algebraic approach. Next, we recall this result. 
\begin{lemma}[List-Decoding of FRS codes \cite{guruswami2013linear}]
\label{lin-alg-list-decoding}
	For the  FRS code $\mathrm{FRS}_q^{(m)}$ of block length $N=n/m$ and rate $R=k/n$, the following holds for all integers $s \in [m]$. Given a received word $\by \in (\mathbb F_q^m)^N$, using $O(n^2+sk^2)$ operations over $\mathbb F_q$, one can find a subspace of dimension at most $s-1$ that contains all encoding polynomials $f \in \mathbb F_q[X]$ of degree less than $k$ whose FRS encoding differs from $\by$ in at most a fraction 
	\begin{align}\label{err-frac}
	\frac{s}{s+1}\left(1-\frac{mR}{m-s+1}\right)
	\end{align}
	of the $N$ codeword positions.
\end{lemma}
\noindent Note that choosing $s=m=1$ in \Lref{lin-alg-list-decoding} corresponds to the unique decoding radius of $(1-R)/2$, while choosing $s \approx 1/\epsilon$ and $m \approx 1/{\epsilon^2}$ ensures a decoding radius of $1-R-\epsilon$. 

We note that the work of  Guruswami and Wang \cite{guruswami2011linear} only considers errors and not symbol erasures. In other words, it is assumed that all coordinates of the received word are either error-free or corrupt and none of them is erased during the transmission. However, it can be observed that the same algorithm also works if $S$ number of coordinates in $\by$ are erased, as provided in the following lemma.

\begin{lemma}
\label{lem_appendix}
    The result of \Lref{lin-alg-list-decoding} still holds if $N$ is replaced by $N-S$ for the case where $S$ out of $N$ symbols are erased.
\end{lemma}
We provide the proof of Lemma \ref{lem_appendix} in the Appendix.

While list-decoding allows for correcting more errors as compared to unique decoding, it is often required to output a unique codeword in certain applications. If the decoder has access to some noise-free side information, then it can use it to prune the list and output a unique codeword. 
This problem was studied in \cite{guruswami2003list} in a communication setup. More specifically, a probabilistic scheme was proposed in \cite{guruswami2003list} in which the transmitter sends a random error-free symbol generated by a random hash function along with this hash function as a side information. The receiver then checks whether there is a unique message in the list that is consistent with the side information or not. If yes, the receiver outputs this message. Otherwise, the receiver declares a decoding failure. 

\noindent Guruswami and Wang \cite{guruswami2013linear} also developed an alternative linear-algebraic list-decoding approach with side information for derivative codes. In this approach, a subspace of candidate polynomials is first determined. The unique message can be then found with high probability by pruning this subspace using the side information. The main advantage of this approach compared to the hashing approach developed in \cite{guruswami2003list} is that the decoder does not need to compute the full list, which may have an exponential size in $s$, and then prune it to get the unique solution. Kopparty \emph{et. al} \cite{kopparty2018improved}  also provide a randomized decoding algorithm for FRS codes that returns a constant-size list which is much smaller than the number of candidate polynomials in the subspace returned by the approach of Gursuwami-Wang described in \Lref{lin-alg-list-decoding}. As a subroutine, one can prune the list by looking at one random symbol of FRS code, i.e., $(f(a),f(a\gamma), \cdots, f(a \gamma^{m-1}))$ for a random $a \in \Fq^*$, and recover the true polynomial with high probability. However, using a similar approach and following the result of \cite[Lemma\,12]{guruswami2016explicit}, one can show that the unique polynomial can be decoded successfully with a probability $p_d$ that is at least 
\be{Wronskin}
p_d \geq 1-\frac{k(s-1)}{q}, 
\ee
using  $s-1$ extra evaluations instead of $m$ evaluations. However, this approach requires the field size to be larger than a certain threshold, i.e., $q>k(s-1)$. To avoid such a constraint for the case where the field size is not relatively \emph{large},  the result of \cite[Lemma\,2]{saraf2011noisy} implies that by using a similar approach of \cite{kopparty2018improved} the unique codeword can be found with probability at least 
\be{Yekhanin}
p_d \geq 1-{t\choose l-1}\left(\frac{k}{n}\right)^{t-l+1},
\ee
 where $l<s$ is the dimension of subspace returned by the list decoding algorithm provided that $t$ extra evaluations $f(a_1), f(a_2) \cdots, f(a_t)$ for randomly picked $a_i \in \Fq^*$ with $a_i\neq a_j$ for all $i,j \in [t]$ are available at the decoder. 

In this work, we provide a new lower bound on the probability of  successful decoding of FRS codes using at most $s-1$ extra evaluations of $f(\cdot)$ and compare our result with the aforementioned existing works in Section\,\ref{sec:listdecoding-results}.

\subsection{Lagrange coded computing}\label{Pre_Lagrange}
Consider a coded computing setup consisting of a master node and a set of $N$ workers. We consider the worst-case adversarial model with up to $A$ computationally-unbounded Byzantine (or malicious) adversarial workers, up to $S$ stragglers and a privacy model where any set of up to $T$ workers can collude. 
Let  $(\bX_1, \cdots, \bX_K)$ denote a batch of $r \times h $ matrices over $\Fq$. The goal is to compute a $D_2$-degree polynomial function $g(\cdot): \Fq^{r \times h} \rightarrow \Fq^{r' \times h'}$ , over this dataset, i.e.,  $g(\bX_i)$  for all $i \in [K]$, where $r'$ and $h'$ are the dimensions of the output matrix. More specifically, we say $g(\cdot)$ is a $D_2$-degree polynomial function if all entries of the output matrix are multivariate polynomial functions of the entries of the input with total degree at most $D_2$, i.e., $\bY=g(\bX)$ implies that 
\be{polynomial-def}
y_{ij}=g_{ij}(x_{11},x_{12},\cdots, x_{rh}),
\ee
where $y_{ij}$ is the $(i,j)$-th entry of $\bY$, for $i \in [r']$ and $j \in [h']$, $x_{lk}$ is the $(l,k)$ entry of $\bX$, for $l \in [r]$ and $k \in [h]$, and, $g_{ij}$ is a multivariate polynomial of total degree at most $D_2$. Let $E\,\deff\, \{ \alpha_1, \cdots, \alpha_N\}$ for some distinct elements $\alpha_1, \cdots, \alpha_N \in \Fq$ and $I\,\deff\,\{\beta_1, \cdots, \beta_{K+T}\}$ for some other distinct elements $\beta_1, \cdots, \beta_{K+T} \in \Fq$. The sets $E$ and $I$ are referred to as the set of evaluation points and  the set of interpolation points, respectively. Note that $E$ and $I$ do not intersect, i.e., $E \cap I=\emptyset$. 
  The underlying encoding polynomial in LCC is the Lagrange interpolation polynomial of degree $D_1=K+T-1$ constructed as
\be{Lagrange-polynomial}
u(z)=\sum_{j=1}^{K} \bX_j \ell_j(z)+\sum_{j=K+1}^{K+T} \bZ_{j} \ell_j(z),
\ee
where $\bZ_j$'s for $j \in \{K+1, \cdots, K+T\}$ are random matrices whose entries are independent and uniformly distributed over $\Fq$ and $\ell_j(z)$'s are called Lagrange monomials specified as
\be{Lagrange-monomials} 
\ell_j(z)=\prod_{l\in [K+T]\setminus \{j\}} \frac{z-\beta_l}{\beta_j-\beta_l},
\ee 
for $j \in [K+T]$. The master node offloads the coded matrix $\tilde{\bX_i}\deff u(\alpha_i)$ to the worker node $i$ whose task is to compute $\tilde{\bY_i}\deff g(u(\alpha_i))=g(\tilde{\bX_i})$. The composed polynomial $f(z)\,\deff\, g(u(z))$  can be recovered provided that at least $D_1 D_2+2A+1$ workers return their computation results, i.e., $\tilde{\bY_i}$'s, to the master node. This can be done by invoking  Berlkamp-Welch (BW) decoder \cite{blahut2008algebraic} individually for all entries of the output matrix of the polynomial function $f(z)$. Note that BW decoder can successfully reconstruct a polynomial of degree $D_1D_2$ provided that $D_1D_2+2A+1$  evaluations of the polynomial are available and  up to $A$ of the evaluations can be erroneous.  Note that since $\bX_i=u(\beta_j)$ for all $j \in [K]$, then $f(\beta_j)=g(u(\beta_j))=g(\bX_j)$.  The master then evaluates $g(\cdot)$ over the interpolation points $\beta_1, \cdots, \beta_K$ to recover the desired computation outcome, i.e., $g(\bX_1), \cdots, g(\bX_K)$. We say that LCC is $S$-resilient and $A$-secure if it is robust against $S$ stragglers and $A$ Byzantine adversaries, respectively, and that it is $T$-private if any set of size up to $T$ workers remain oblivious to the content of dataset. 
It is shown in \cite{yu2019lagrange} that the number of required workers for an $S$-resilient, $A$-secure, and $T$-private LCC to compute $\{g(\bX_i)\}_{i=1}^{K}$ for a $D_2$-degree polynomial $g(\cdot)$ is lower bounded as
\be{Lagrange-lowerbound}
N \geq (K+T-1)D_2+S+2A+1,
\ee
where $ D=(K+T-1)D_2$ is the degree of the composed polynomial $f(\cdot)$ to be interpolated at the master node during the decoding step. The lower bound provided in \eqref{Lagrange-lowerbound} implies that tolerating Byzantine adversaries in LCC is twice as costly as stragglers, i.e., the additional number of workers required to tolerate each Byzantine worker is equal to what is needed to tolerate two stragglers.

\section{System Model}
\label{Section: System Model}
We consider a coded computing setup consisting of a master node and a set of $N$ workers. The goal is to design a coded computing scheme that is $S$-resilient, $A$-secure and $T$-private, where we consider the worst-case adversarial model with up to $A$ computationally-unbounded Byzantine adversarial workers. Unlike LCC \cite{yu2019lagrange}, where the master node does not do any computations except encoding and decoding, we allow the master node to do a tiny amount of extra computations that is negligible compared to the computations of each worker. Specifically, the normalized extra computation cost of the master node with respect to the workload of each worker must go to $0$ as the workload of each worker grows.  

A well-known class of coded computing schemes  extensively studied in the literature for this setup employs polynomial evaluations to encode data. We refer to them as \emph{polynomial-based} coded computing schemes. In such schemes, the shares sent to the worker nodes are evaluations of a certain polynomial over a finite field $\Fq$. The worker nodes perform a predefined computation task over their share(s), e.g., polynomial evaluation, matrix multiplication, etc., and return the results to the master node. The master node then follows a decoding process involving a polynomial interpolation to recover the overall computation outcome. We denote the $D$-degree polynomial to be interpolated at the decoding step by $f(\cdot)$.
 Such a class of coded computing schemes includes, but is not limited to, LCC \cite{yu2019lagrange}, polynomial codes \cite{yu2017polynomial} and MatDot codes \cite{dutta2019optimal}.
 It is well-known that the polynomial $f(\cdot)$ can be uniquely recovered provided that up to $(N-D-1)/2$ evaluations, out of $N$ available evaluation points, are erroneous. This can be done by utilizing efficient Reed-Solomon decoding algorithms at the master node. This, in a high level, imposes a limit on the maximum number of Byzantine workers that can be tolerated in prior works on polynomial-based coded computing. In this work, we break this barrier by employing folded Reed-Solomon (FRS) codes instead of RS codes often used in the polynomial-based schemes together with leveraging their list-decoding algorithms instead of the unique decoding algorithm for RS codes. The output of the list decoder is then pruned using a tiny amount of side information to recover the unique computation outcome. The computation cost of this side information is negligible compared to the computation cost of each worker. This will be clarified more later in Section\,\ref{Section: Folded Lagrange Coded Computing}. 

 In order to illustrate the key idea of our method more explicitly, we consider the LCC scheme \cite{yu2019lagrange} described earlier in Section\,\ref{Pre_Lagrange}.  
We propose a variant of LCC, referred to as folded LCC (FLCC), inspired by the folded RS code construction. Consider $m$ batches of size $K$ which can be considered as a larger batch of size $k \,\deff\, mK$, i.e., $(\bX_1, \cdots, \bX_{k})$, where the goal is to compute $g(\bX_i)$ for all $i \in [k]$. The parameter $m$ is an arbitrary integer, referred to as the folding parameter.
We also assume $ N= (K+T-1)D_2+S+2A+1$ which is the minimum number of workers required in LCC. Let $E_m \,\deff\,\{\alpha_i=\alpha^{i-1} |i \in [N]\}$ denote the set of evaluation points in our proposed scheme, where $\alpha$ is a primitive element of $\Fq$ and $q > N$.
Note that the evaluation points here are picked more specifically compared to those of LCC, but that does not impose any limitations on the scheme. Furthermore, instead of the polynomial in \eqref{Lagrange-polynomial}, we construct the following encoding polynomial
 \be{folded-Lagrange-polynomial}
u_m(z)=\sum_{j=1}^{Km} \bX_j \ell_j(z)+\sum_{j=Km+1}^{(K+T)m} \bZ_{j} \ell_j(z),
\ee
where $\bZ_j$'s for $j \in \{Km+1, \cdots, (K+T)m\}$ are random matrices whose entries are independent and uniformly distributed over $\Fq$ and $\ell_j(z)$'s are Lagrange monomials defined as
\be{Lagrange-monomials} 
\ell_j(z)=\prod_{l\in [(K+T)m]\setminus j} \frac{z-\beta_l}{\beta_j-\beta_l},
\ee
where $I_m\,\deff\, \{\beta_1, \cdots, \beta_{m(K+T)}\}$ for distinct $\beta_i$'s  belonging to $\Fq$ such that $E_m \cap I_m= \emptyset$. We refer to this scheme as FLCC. In FLCC, the share of encoded data sent to the worker $i$ consists of the evaluations of $u_m(\cdot)$ over the points $\alpha_{m(i-1)+1} \cdots, \alpha_{mi} $, i.e., $u_m(\alpha_{m(i-1)+1}), \cdots, u_m(\alpha_{mi})$. Let $f_m(z) \,\deff\, g(u_m(z))$ denote the composed polynomial to be interpolated at the decoder in our scheme. The task of each worker node is then to compute $g(\cdot)$ on all of its associated evaluations separately, i.e., node $i$ computes $f_m(\alpha_{m(i-1)+1}), \cdots, f_m(\alpha_{mi})$ and returns the results to the master node. 

Intuitively speaking, the encoding polynomial considered in our scheme is similar to that of LCC in which $m$ batches of data are regarded as a single dataset of size $m$-times larger. In other words, in the encoding step of our protocol, we first encode the data according to the RS encoder and the coded symbols are then \emph{folded} with parameter $m$, resembling the FRS encoding procedure. Consequently, we can apply list-decoding algorithms developed for FRS in the literature to gain a better resiliency-security-privacy trade-off in FLCC compared to that of LCC. A linear-algebraic list-decoding algorithm for FRS codes is considered in the next section along with our proposed methods to improve it when certain side information is available at decoder. We then discuss in Section\,\ref{Section: Folded Lagrange Coded Computing} how this result can be utilized to improve upon the performance of LCC decoder in terms of the achievable triples of $(S, A, T)$.

\section{List decoding FRS codes with side information}\label{sec:listdecoding-results}

In this section, we provide our approach to adapt list-decoding techniques to make the polynomial-based coded computing protocols more \emph{robust} against malicious adversaries. This is done in such a way that it can be used as a \textit{black box} and regardless of the technical details associated with the encoder and decoder of the underlying coded computing scheme. We then illustrate in Section \ref{Section: Folded Lagrange Coded Computing} how the proposed techniques can be applied to LCC.
\\
In particular, we consider the linear-algebraic FRS list decoder introduced in \cite{guruswami2011linear}. Let $\mathcal{W}$ denote the linear space of univariate polynomials of degree at most $k-1$ over $\Fq$ and $\mathcal{L}$ denote the list of candidate polynomials at the output of the list decoder that can be represented by an affine subspace $U$. The elements of $U$ can be represented as $\bff=\bM\bx+\bz$ for $\bx\in \Fq^l$, where $l <s$, $\bM \in \Fq^{k\times l} $ and $\bz \in \Fq^k$. The vector $\bff =(f_0, f_1, \cdots, f_{k-1})^{\mathrm T}$ denotes the coefficients of the corresponding polynomial in $U$. It is also shown in \cite{guruswami2011linear} that $\bM$ can be assumed to have $l \times l$ identity matrix $\bI_{l}$ as a submatrix, without any extra computation. Note that the location of the identity submatrix is not known prior to applying the list-decoding algorithm at the decoder. 

Now, suppose that the decoder of the FRS code can request to have access to $l<s$ additional error-free evaluations of $f(\cdot)$ as a side information. This can be done with the aim of pruning the output of the list-decoding algorithm specified in \Lref{lin-alg-list-decoding} to uniquely recover $f(\cdot)$. In the next theorem, we provide a result that this is always possible provided that the $l<s$ evaluation points can be decided after applying the list-decoding algorithm to the received word $\by$. 

\begin{theorem}\label{main-result}
	For the FRS code of length $N=\frac{n}{m}$ and rate $R=\frac{k}{n}$ and for all $s \in [m]$, the polynomial $f(\cdot)$ can be uniquely recovered if 
	\begin{enumerate}[i)]
		\item the received word $\by \in (\Fq^m)^N$ differs from the FRS codeword corresponding to $f(\cdot)$ in at most $\frac{s}{s+1}(1-\frac{mR}{m-s+1})$ fraction of the $N$ symbols and 
		\item up to $s-1$ additional evaluations of $f(\cdot)$ can be requested and are provided error-free, and assuming that the corresponding evaluation points can be decided after $\by$ is received. 
	\end{enumerate}
	Moreover, the entire algorithm is run with $O(n^2+sk^2)$ complexity. 
\end{theorem}

\begin{proof}
	Let $\tilde{\bV} \in \Fq^{n \times k}$ be an arbitrary Vandermonde matrix. Let also the affine subspace $U$ that contains the true polynomial be represented by $\bff=\bM\bx+\bz$ for $\bx\in \Fq^l$, where $l <s$, $\bM \in \Fq^{k\times l} $ and $\bz \in \Fq^k$, as discussed above.  Then, $\tilde{\bV}\bM$ is full rank since both $\tilde{\bV}$ and $\bM$ are full rank and $n \geq k \geq l$. By using the fast multiplication algorithm available for Vandermonde matrices \cite{pan2001structured}, $\tilde{\bV}\bM$ can be formed in $O(nl\log n)$. Furthermore, an $l \times l$ full-rank submatrix of $\tilde{\bV}\bM$  can be found in time $O(nl^2)$ by Gaussian elimination. Then, the rows in $\tilde{\bV}$ associated to this $l\times l$ submatrix form an $l \times k$ Vandermonde matrix, namely,
	\be{Vandermonde-deff}
	\bV\,\deff\, \begin{bmatrix}
		1 & \lambda_1  & \cdots & \lambda_1^{k-1}\\
		1 & \lambda_2  & \cdots & \lambda_2^{k-1}\\
		\vdots & \vdots & \vdots & \vdots \\
		1 & \lambda_l  & \cdots & \lambda_l^{k-1}
	\end{bmatrix},
	\ee
	such that $\tilde{\bV}\bM$ is full rank. The evaluations of $f(\cdot)$ over the points $\lambda_1, \cdots, \lambda_l$ associated to the Vandermonde matrix suffice to uniquely recover $f(\cdot)$.  To see that, let $\by_e\,\deff\, \bV \bff$ denote the vector of extra evaluations of $f(\cdot)$ over $\lambda_1, \cdots, \lambda_l$. 
	Therefore, one can write 
	\be{extra-eqs}
	\by_e=\bV\bM\bx+\bV\bz.
	\ee
	This implies that if $\bV\bM$ is full rank, then \eqref{extra-eqs} can be solved for $\bx$ which is then utilized to determine $\bff$, thereby uniquely recovering the polynomial $f(\cdot)$. The computational complexity of the entire algorithm is dominated by that of the list-decoding algorithm described in \Lref{lin-alg-list-decoding} which is $O(n^2+sk^2)$ by noting that $l <s$. 
\end{proof}

The result of \Tref{main-result} implies that if the set of extra evaluation points $\Lambda\deff \{\lambda_1, \cdots, \lambda_l\}$ can be decided by the decoder after observing the entire received vector, the output subspace provided by the list-decoding algorithm, as specified in \Lref{lin-alg-list-decoding}, can be efficiently pruned to uniquely recover the polynomial $f(\cdot)$ in a deterministic fashion. 
This provides a pruning algorithm with a computational complexity that is dominated by that of the corresponding list-decoding algorithm. Note that in a communication setting, with the encoder and the decoder being separate entities, such an assumption on the side information would necessitate multiple rounds of communication that may not be desirable in practice. However, in coded computing settings, the encoding and decoding are both done by the same entity, i.e., the master node. Hence, obtaining error-free side information after the results are received from the workers does not impose any major hurdle to the protocol. This comes only at the cost of extra computation complexity and latency, which can be characterized and optimized based on the limitations of the master node.

In order to mitigate the latency of computing the side information, we propose an alternative probabilistic algorithm. In this algorithm, the master node does not have to wait till the results are received from the workers and can compute the side information in parallel to them. This is described in the following theorem. It is shown that the polynomial $f(\cdot)$ can be uniquely recovered using this algorithm with \emph{high probability} if the size of the underlying finite field $\Fq$ is large enough. Moreover, if the unique recovery is not possible, then the decoder can identify it as a decoding failure, i.e., the output of this decoding scheme with the probabilistic pruning algorithm is either the true outcome or a decoding failure, provided that the number of errors is bounded by a certain threshold.

To provide our result for the probabilistic scheme, we need the following definition and Lemma that are provided below.

\begin{definition}
	A vector $\bv \in \Fq^k $ is a called a Vandermonde-type vector, or a V-vector in short, if $\bv=(1,\lambda, \lambda^2,\cdots, \lambda^{k-1})^{\mathrm T} $ for some $\lambda \in \Fq \setminus \{0\}$. 
\end{definition}

The following lemma is used to prove our main result in this section. 
\begin{lemma}\label{lem_v}
	For any arbitrary matrix $\bA_{k\times l}\ (k>l)$ over $\Fq$ of rank $r\leq l$, there exist at most $r$ distinct V-vectors of length $k$ that lie in the column space of $\bA$.  
\end{lemma}
\begin{proof}
	Assume to the contrary that there exist $r+1$ distinct V-vectors that lie in the subspace spanned by the columns of $\bA$, namely, $\bv_1, \cdots,\bv_{r+1}$. Let 
	$
	\bV_{k\times(r+1)}\,\deff\, \left[\bv_1| \bv_2| \hdots| \bv_{r+1} \right] 
	$, which is a Vandermonde matrix. Then, the column space of $\bV$ is a subspace of the linear space spanned by the columns of $\bA$. On the other hand, it is well-known that a Vandermonde matrix with distinct columns is always full rank, i.e., the column space of $\bV$ has dimension $r+1$ while $\text{rank}(A)=r$ which is a contradiction.
\end{proof}

\begin{theorem}\label{probability}
	For the FRS code of length $N=\frac{n}{m}$ and rate $R=\frac{k}{n}$ and for any $s \in [m]$, the  polynomial $f(\cdot)$ can be uniquely recovered with probability 
	 \be{Our_prob} p_d 
	  \geq \frac{\sum\limits_{i=l}^{t}{q-k+l-1\choose i}{k-l\choose t-i}}{{q-1\choose t}},
	 \ee 
	 where $l<s$ is the dimension of subspace returned by the list decoding algorithm, 
	 provided that
	\begin{enumerate}[i)]
		\item the received word $\by \in (\Fq^m)^N$ differs from the FRS codeword corresponding to $f(\cdot)$ in at most $\frac{s}{s+1}(1-\frac{mR}{m-s+1})$ fraction of $N$ symbols, and
		\item  the side information $f(\lambda_1), \cdots, f(\lambda_{t})$ are given to the decoder, where $\lambda_i$ is drawn uniformly at random from $\Fq \backslash\{0, \lambda_1, \cdots, \lambda_{i-1}\}$, for $i \in [t]$. 
	\end{enumerate}	 
	Furthermore, if the unique recovery is not possible, the decoder can identify it.  \end{theorem}
\begin{proof}
		Recall that by applying the linear-algebraic list-decoding algorithm, discussed in \Lref{lin-alg-list-decoding}, we have $\bff=\bM\bx+\bz$ for which $\bM$ has $\bI_l$ as its submatrix for some $l< s$. Then, without loss of generality, we can assume 
	\be{M-structure}
	\bM=\begin{bmatrix}
		\begin{array}{c}
			\tilde{\bM}_{(k-l)\times l }\\
			\hline
			\bI_{l }
		\end{array}
	\end{bmatrix},
	\ee
	where $\tilde{\bM}$ is an $(k-l)\times l$ matrix over $\Fq$ since the locus of identity submatrix is determined after the list-decoding algorithm is applied. Let $\bV$ and $\by_e$ be as characterized in \eqref{Vandermonde-deff} and \eqref{extra-eqs}, respectively. 

	Then, if $\bV\bM$ is full rank, \eqref{extra-eqs} can be solved for $\bx$ which is then utilized to determine $\bff$, thereby uniquely recovering the polynomial $f(\cdot)$. Furthermore, note that all the following are equivalent:
	\begin{align}
	\bV\bM \quad \text{is full rank.} &\Longleftrightarrow \quad \bM^{\mathrm T}\bV^{\mathrm T}\quad \text{ is full rank.}\\ & \Longleftrightarrow \quad \Span {\bV^{\mathrm T}} \cap \Span{\bN}=\{\boldsymbol{0}\},\label{equivalent}
	\end{align}
	where $\bN_{k \times (k-l)}$ is a matrix whose columns span the null-space of $\bM^{\mathrm T}_{l\times k }$, i.e., $\bM^{\mathrm T}\bN=\boldsymbol{0}$. Moreover, note that $\Span {\bV^{\mathrm T}} \cap\Span{\bN}=\{\boldsymbol{0}\}$, where $\boldsymbol{0}$ denotes the all-zero vector, if and only if $\left [ \bN_{k\times (k- l)} | \bV_{k\times l} \right ]$ is full rank. 
	According to the result of \Lref{lem_v}, there  exist at most $k-l$ V-vectors of length $k$ that lie in $\Span{\bN}$.  The  field elements that correspond to these vectors are referred to as \textit{bad} choices for $\lambda_i$'s. Hence, there are at most $k-l$ bad choices for $\lambda_i$'s in $\Fq^*$. Suppose that we randomly pick $t$ distinct $\lambda_i$'s from $\Fq^*$. The number of cases where at least $l$ out of $t$ random choices of $\lambda_i$'s do not fall into this set of size at most $k-l$ is equal to $\sum\limits_{i=l}^{t}{q-1-k+l\choose i}{k-l \choose t-i}$. In such cases, we set $\bV$ to be the corresponding Vandermonde matrix and $[\bN_{k\times (k-l)}|\bV]$ would be full rank.  There are ${q-1\choose t}$ different choices of $t$-subsets of $\Fq^*$. 
 Consequently, the probability of $[\bN_{k\times (k-l)}|\bV]$ being full rank   is at least 
 $\frac{\sum\limits_{i=l}^{t}{q-k+l-1\choose i}{k-l\choose t-i}}{{q-1\choose t}}$ 
 by
 picking $t$ extra evaluation points uniformly at random from $\Fq^*$.

	Note that in the event that $[\bN_{k\times (k-l)}|\bV]$ is not full-rank, the system of linear equations characterized in \eqref{extra-eqs} does not have a unique solution. In such a case, the decoder verifies the unique recovery is not possible and declares recovery failure.

\end{proof}

We now compare the lower bound on $p_d$ provided in  \Tref{probability}  with \eqref{Wronskin} and \eqref{Yekhanin} that  follow from  \cite[Lemma\,12]{guruswami2016explicit} and {\cite[Lemma\,2]{saraf2011noisy}}, respectively, as described in Section\,\ref{Pre_FRS}. The lower bound provided in \eqref{Wronskin} is compared with our result in Figure\,\ref{WronskinVSus} for $k=1000$ and $k=10000$. We set $t=s-1$ in our proposed scheme to have a fair comparison with the lower bound provided in \eqref{Wronskin}. For $k=1000$, our result is almost the same as what is guaranteed by \eqref{Wronskin} as $k(s-1)$ is well below the field size $q=10003$. For the case where $k=10000$, the lower bound in \eqref{Wronskin} is trivial for $s\geq10$ where our result still provides a non-trivial bound on $p_d$. In other words, the lower bound characterized in \Tref{probability} is non-trivial for a wider range of encoding parameters. Moreover, our result allows having $t\geq s$ extra evaluation points and the lower bound on $p_d$ provided in \eqref{Our_prob} improves as the number of extra evaluations $t$ increases. Consequently, our result does not necessarily require $q\gg k(s-1)$ as one can improve the lower bound on $p_d$ with a few more extra evaluations. For instance, with $t=s+3$, i.e., having $4$ more evaluations, the lower bound provided in \Tref{probability} on $p_d$ significantly improves as demonstrated in Figure\,\ref{WronskinVSus}. This illustrates the superiority of our result over the one that  is established based on the result in  \cite[Lemma\,12]{guruswami2016explicit},  provided in \eqref{Wronskin}. Roughly speaking, the advantage of our approach over the one  established based on the result in  \cite[Lemma\,12]{guruswami2016explicit} is due to  the difference in how the extra evaluation points are picked. In our approach, we pick all extra evaluation points independently and uniformly at random from $\Fq^*$  where the extra evaluation points in the latter approach keeps the same structure as FRS symbols, i.e., they are equal to $a, a\gamma, \cdots, a\gamma^{s-1}$, for some random $a\in \Fq^*$. In a sense, we require more randomness than what is needed in the latter approach.   
\begin{figure*}[htb!]
    \begin{minipage}{0.45\textwidth}
        \centering
%

\begin{tikzpicture}

\definecolor{mycolor1}{rgb}{0.15,0.15,0.15}
\definecolor{mycolor2}{rgb}{0,0,1}
\definecolor{mycolor3}{rgb}{1,0,0}
\definecolor{mycolor4}{rgb}{0,0.45,0.74}
\definecolor{mycolor5}{rgb}{0.64,0.08,0.18}
\definecolor{mycolor6}{rgb}{0.2,0.6,0.2}
\definecolor{mycolor7}{rgb}{1,0 ,1}

\begin{axis}[
 x label style={at={(axis description cs:0.47,0)},anchor=north},
    y label style={at={(axis description cs:0.05,0.5)},rotate=0,anchor=south},
    xlabel={$s$},
    ylabel={Lower bound on $p_d$},
scale only axis,
every outer x axis line/.append style={mycolor1},
every x tick label/.append style={font=\color{mycolor1}},
every outer y axis line/.append style={mycolor1},
every y tick label/.append style={font=\color{mycolor1}},
width=2.75in,
height=2in,
xmin=0, xmax=31.0000,
ymin=0, ymax=1,
axis on top,
legend entries={\footnotesize{ \cite{guruswami2016explicit}, $k=1000$,}, \footnotesize{Ours, $k=1000$, $t=s$}, \footnotesize{Ours, $k=1000$, $t=s+4$,}, \footnotesize{\cite{guruswami2016explicit}, $k=10000$,},
\footnotesize{Ours, $k=10000$, $t=s$,},\footnotesize{Ours, $k=10000$, $t=s+4$.}},
legend style={ nodes={scale=.7, transform shape}, legend columns=2},
legend cell align=right,
legend style={at={(0.02,1.15)},anchor=west},
mark options={solid,scale=1.3}
]
\addplot [
color=mycolor2,
dashed,
mark=x,
line width=1.0pt,
]
coordinates{
 (1,0.99)
 (2,0.980001)
 (3,0.970001)
 (4,0.960001)
 (5,0.950001)
 (6,0.940002)
 (7,0.930002)
 (8,0.920002)
 (9,0.910003)
 (10,0.900003)
 (11,0.890003)
 (12,0.880004)
 (13,0.870004)
 (14,0.860004)
 (15,0.850004)
 (16,0.840005)
 (17,0.830005)
 (18,0.820005)
 (19,0.810006)
 (20,0.800006)
 (21,0.790006)
 (22,0.780007)
 (23,0.770007)
 (24,0.760007)
 (25,0.750007)
 (26,0.740008)
 (27,0.730008)
 (28,0.720008)
 (29,0.710009)
 (30,0.700009)

};
\addplot [
color=mycolor3,
dashed,
mark=o,
line width=1.0pt,
]
coordinates{
 (1,0.99)
 (2,0.9801)
 (3,0.970299)
 (4,0.960596)
 (5,0.95099)
 (6,0.94148)
 (7,0.932065)
 (8,0.922744)
 (9,0.913516)
 (10,0.90438)
 (11,0.895335)
 (12,0.886381)
 (13,0.877516)
 (14,0.86874)
 (15,0.860052)
 (16,0.85145)
 (17,0.842935)
 (18,0.834504)
 (19,0.826158)
 (20,0.817895)
 (21,0.809714)
 (22,0.801615)
 (23,0.793598)
 (24,0.78566)
 (25,0.777802)
 (26,0.770022)
 (27,0.76232)
 (28,0.754695)
 (29,0.747146)
 (30,0.739672)
};

\addplot [
color=mycolor5,
dashed,
mark=diamond,
line width=1.0pt,
]
coordinates{
 (1,1)
 (2,0.99994)
 (3,0.99986)
 (4,0.99976)
 (5,0.99964)
 (6,0.9995)
 (7,0.99934)
 (8,0.99916)
 (9,0.99896)
 (10,0.998741)
 (11,0.998501)
 (12,0.998241)
 (13,0.997961)
 (14,0.997662)
 (15,0.997342)
 (16,0.997003)
 (17,0.996643)
 (18,0.996264)
 (19,0.995865)
 (20,0.995446)
 (21,0.995007)
 (22,0.994549)
 (23,0.99407)
 (24,0.993572)
 (25,0.993054)
 (26,0.992516)
 (27,0.991958)
 (28,0.991381)
 (29,0.990784)
 (30,0.990167)

};

\addplot [
color=mycolor4,
dashed,
mark=square,
line width=1.0pt,
]
coordinates{
 (1,0.900003)
 (2,0.800006)
 (3,0.700009)
 (4,0.600012)
 (5,0.500015)
 (6,0.400018)
 (7,0.300021)
 (8,0.200024)
 (9,0.100027)
 (10,2.99991e-05)

};
\addplot [
color=mycolor5,
dashed,
mark=star,
line width=1.0pt,
]
coordinates{
 (1,0.900002)
 (2,0.810003)
 (3,0.729002)
 (4,0.656101)
 (5,0.59049)
 (6,0.531439)
 (7,0.478293)
 (8,0.430461)
 (9,0.387413)
 (10,0.348669)
 (11,0.313799)
 (12,0.282416)
 (13,0.254172)
 (14,0.228752)
 (15,0.205874)
 (16,0.185284)
 (17,0.166753)
 (18,0.150075)
 (19,0.135065)
 (20,0.121556)
 (21,0.109399)
 (22,0.0984566)
 (23,0.088609)
 (24,0.0797462)
 (25,0.0717699)
 (26,0.0645912)
 (27,0.0581305)
 (28,0.052316)
 (29,0.0470831)
 (30,0.0423735)

};
\addplot [
color=mycolor6,
dashed,
mark=triangle,
line width=1.0pt,
]
coordinates{
 (1,0.99999)
 (2,0.999885)
 (3,0.999684)
 (4,0.999329)
 (5,0.998752)
 (6,0.997871)
 (7,0.9966)
 (8,0.99485)
 (9,0.992531)
 (10,0.989559)
 (11,0.985854)
 (12,0.981345)
 (13,0.975969)
 (14,0.969676)
 (15,0.962427)
 (16,0.954193)
 (17,0.94496)
 (18,0.934724)
 (19,0.923491)
 (20,0.91128)
 (21,0.898116)
 (22,0.884037)
 (23,0.869085)
 (24,0.853311)
 (25,0.83677)
 (26,0.819522)
 (27,0.801631)
 (28,0.783164)
 (29,0.764187)
 (30,0.744772)

};

\end{axis}

\end{tikzpicture}
        \caption{ \small Comparison of the lower bounds provided in \eqref{Wronskin} and \eqref{Our_prob} for $k=1000$ and $k=10000$. The field size is $q=100003$. Note that in both case, our proposed lower bound  improves significantly by adding $4$ more evaluations, i.e., $t=s+3$.    }\label{WronskinVSus}
    \end{minipage}\hfill
    \begin{minipage}{0.47\textwidth}
    \vspace{2.7mm}
        \centering
%

\begin{tikzpicture}

\definecolor{mycolor1}{rgb}{0.15,0.15,0.15}
\definecolor{mycolor2}{rgb}{0,0,1}
\definecolor{mycolor3}{rgb}{1,0,0}
\definecolor{mycolor4}{rgb}{0,0.45,0.74}
\definecolor{mycolor5}{rgb}{0.64,0.08,0.18}
\definecolor{mycolor6}{rgb}{0.2,0.6,0.2}
\definecolor{mycolor7}{rgb}{1,0 ,1}

\begin{axis}[
 x label style={at={(axis description cs:0.47,0.0)},anchor=north},
    y label style={at={(axis description cs:0.05,0.5)},rotate=0,anchor=south},
    xlabel={$t$},
    ylabel={Lower bound on $p_d$},
scale only axis,
every outer x axis line/.append style={mycolor1},
every x tick label/.append style={font=\color{mycolor1}},
every outer y axis line/.append style={mycolor1},
every y tick label/.append style={font=\color{mycolor1}},
width=2.75in,
height=2in,
xmin=10, xmax=70,
ymin=0, ymax=1,
axis on top,
legend entries={\footnotesize{\cite{saraf2011noisy}}, \footnotesize{Ours, q=100003,},\footnotesize{Ours, q=10007.} },
legend style={ nodes={scale=.9, transform shape}, legend columns=3},
legend cell align=right,
legend style={at={(0.01,1.12)},anchor=west},
mark options={solid,scale=1.3}
]
\addplot [
color=mycolor6,
dashed,
mark=o,
line width=1.0pt,
]
coordinates{
 (35,0)
 (36,0.298578)
 (37,0.53656)
 (38,0.696367)
 (39,0.802639)
 (40,0.87267)
 (41,0.918429)
 (42,0.948091)
 (43,0.967175)
 (44,0.979367)
 (45,0.987105)
 (46,0.991984)
 (47,0.995043)
 (48,0.996949)
 (49,0.998131)
 (50,0.998861)
 (51,0.999308)
 (52,0.999582)
 (53,0.999748)
 (54,0.999849)
 (55,0.99991)
 (56,0.999946)
 (57,0.999968)
 (58,0.999981)
 (59,0.999989)
 (60,0.999993)
 (61,0.999996)
 (62,0.999998)
 (63,0.999999)
 (64,0.999999)
 (65,1)
 (66,1)
 (67,1)
 (68,1)
 (69,1)
 (70,1)

};



\addplot [
color=mycolor3,
dashed,
mark=x,
line width=1.0pt
]
coordinates{
 (10,0.905294)
 (11,0.994925)
 (12,0.999801)
 (13,0.999994)
 (14,1)
 (15,1)
 (16,1)
 (17,1)
 (18,1)
 (19,1)
 (20,1)
 (21,1)
 (22,1)
 (23,1)
 (24,1)
 (25,1)
 (26,1)
 (27,1)
 (28,1)
 (29,1)
 (30,1)
 (31,1)
 (32,1)
 (33,1)
 (34,1)
 (35,1)
 (36,1)
 (37,1)
 (38,1)
 (39,1)
 (40,1)
 (41,1)
 (42,1)
 (43,1)
 (44,1)
 (45,1)
 (46,1)
 (47,1)
 (48,1)
 (49,1)
 (50,1)
 (51,1)
 (52,1)
 (53,1)
 (54,1)
 (55,1)
 (56,1)
 (57,1)
 (58,1)
 (59,1)
 (60,1)
 (61,1)
 (62,1)
 (63,1)
 (64,1)
 (65,1)
 (66,1)
 (67,1)
 (68,1)
 (69,1)
 (70,1)

};

\addplot [
color=mycolor4,
dashed,
mark=star,
line width=1.0pt
]
coordinates{
 (10,0.35263)
 (11,0.701874)
 (12,0.89194)
 (13,0.967099)
 (14,0.991225)
 (15,0.997891)
 (16,0.999534)
 (17,0.999904)
 (18,0.999981)
 (19,0.999997)
 (20,0.999999)
 (21,1)
 (22,1)
 (23,1)
 (24,1)
 (25,1)
 (26,1)
 (27,1)
 (28,1)
 (29,1)
 (30,1)
 (31,1)
 (32,1)
 (33,1)
 (34,1)
 (35,1)
 (36,1)
 (37,1)
 (38,1)
 (39,1)
 (40,1)
 (41,1)
 (42,1)
 (43,1)
 (44,1)
 (45,1)
 (46,1)
 (47,1)
 (48,1)
 (49,1)
 (50,1)
 (51,1)
 (52,1)
 (53,1)
 (54,1)
 (55,1)
 (56,1)
 (57,1)
 (58,1)
 (59,1)
 (60,1)
 (61,1)
 (62,1)
 (63,1)
 (64,1)
 (65,1)
 (66,1)
 (67,1)
 (68,1)
 (69,1)
 (70,1)

};

\end{axis}

\end{tikzpicture}
\caption{\small  Comparison of the lower bounds provided in \eqref{Yekhanin} and \eqref{Our_prob}. The plots show the lower bounds on $p_d$ versus the number of extra evaluation points available at the decoder for $s=10$. Other parameters are $q=10007, 100003$ and $n=2000$.}\label{YekhaninVSus}
    \end{minipage}
\end{figure*}

As discussed in Section\,\ref{Pre_FRS}, the lower bound on $p_d$ provided in \eqref{Yekhanin} that follows from the results in {\cite[Lemma\,2]{saraf2011noisy}} is also improved when more than $s-1$ extra evaluations are available. Consequently, the large field size is not required for the lower bound on $p_d$ characterized in $\eqref{Yekhanin}$ as in our result in \Tref{probability}. In Figure\,\ref{YekhaninVSus}, we compare the lower bound provided in \eqref{Yekhanin} with our result for two different field sizes $q=10007, 100003$. The bounds are plotted versus $t$, the number of extra evaluations available at the decoder, for $s=10$. It is illustrated that our result requires a significantly smaller number of extra evaluations to guarantee a reasonable successful decoding probability $p_d$. The advantage of our approach might be justified by noting that the result in {\cite[Lemma\,2]{saraf2011noisy}} is valid for any  linear code and not necessarily FRS codes. Intuitively, we leverage the certain structure of the encoding matrix of FRS codes to arrive at a better bound on the decoding probability compared to the one provided in {\cite[Lemma\,2]{saraf2011noisy}} for a general linear code.

\noindent 

\section{Folded Lagrange Coded Computing}

\label{Section: Folded Lagrange Coded Computing}
In this section, we demonstrate how the FRS list-decoding algorithm together with the pruning algorithms proposed in Section\,\ref{sec:listdecoding-results} can be utilized to break the barrier on the number of Byzantine workers that can be tolerated in LCC.

Let the parameters $m$, $N$ and $K$ be associated with our proposed FLCC, specified in Section\,\ref{Section: System Model}. The main result of this section is that the lower bound on $N$ in FLCC can be well-approximated by $(K+T)D_2+A+S-1$ for sufficiently large, but fixed, folding parameter $m$. This implies that Byzantine adversaries are \emph{as costly as} stragglers in terms of the number of additional workers required in FLCC reducing their effect by a factor of $2$ compared to LCC. 

 Similar to LCC, the decoding algorithm is performed over the received matrices element-wise. In theory, one can apply the decoding algorithm for FRS codes with side information as discussed in Section\,\ref{sec:listdecoding-results} individually for all $r'\times h'$ elements of the returned matrices. In the rest of this Section, we assume that all the steps discussed below are performed element-wise on the entries of   $\tilde{\bY_i}$'s individually. In practice, one might be able to perform the algorithm collaboratively on all entries at once and amortize the computational cost over all entries. A similar approach is  provided recently in  \cite{subramaniam2019collaborative} to decode polynomial codes by utilizing a collaborative decoder for interleaved generalized RS codes. It is shown that up to $N-K-1$ Byzantine workers can be tolerated under the additive Gaussian model with \emph{high} probability. The main difference between our proposed method and that of \cite{subramaniam2019collaborative} is that we do not make assumptions on the error model and its probability distribution. That is, we consider the worst-case adversarial model.

  In FLCC, the master node first finds the linear subspace $U$ of dimension at most $s-1$ containing the legitimate polynomial by applying the linear-algebraic list-decoding algorithm. Then it  determines the extra evaluation points needed to uniquely identify $f(\cdot)$ in $U$ according to the procedure described in the proof of \Tref{main-result}. The master node then performs extra computations to evaluate $f(\cdot)$ over these points which results in uniquely determining $f(\cdot)$ by solving the system of linear equations in \eqref{extra-eqs}. Let 
$$
r\,\deff\,\frac{(K+T-\frac{1}{m})D_2+1}{N-S},
$$
which is referred to as the \emph{modified rate} of FLCC. The fraction of adversaries tolerated in FLCC is characterized in the following theorem. 

\begin{theorem}\label{FLCC-decoder-nodes}
An FLCC with folding parameter $m$ and the number of worker nodes equal to $N$ is $S$-resilient, $A$-secure and $T$-private to compute $\{g(\bX_i)\}_{i=1}^{mK}$ for a $D_2$-degree polynomial $g(\cdot)$ as long as 
\be{FLCC-bound}
\frac{A}{N-S} \leq \frac{s^*}{s^*+1}(1-\frac{mr}{m-s^*+1}),
\ee
 where $s^*$ is equal to $\ceil{\tilde{s}}-1$ or $\ceil{\tilde{s}}$, where $\tilde{s}=	\frac{\sqrt{m(m+1)(m(1-r)+2)r}-(m+1)}{mr-1}$, depending on which one results in a larger RHS in \eqref{FLCC-bound}. 

\end{theorem}  
\begin{proof}
	 We first note that the degree of the composed polynomial interpolated at FLCC decoder is $(m(K+T)-1)D_2$.
	 By the result of \Lref{lem_appendix}, we can use the result of \Lref{lin-alg-list-decoding} with $S$ straggling nodes by replacing $N$ with $N-S$, i.e., replacing $R$ in \eqref{err-frac} by the modified rate $r$. 
	 Let $a(s)\,\deff\, \frac{s}{s+1}(1-\frac{mr}{m-s+1})$. Then, the result of \Tref{main-result} implies that FLCC is $S$-resilient and $A$-secure as long as 
	 \be{frac-upperbound}
	 \frac{A}{N-S} \leq a(s),
	 \ee
	 for any arbitrary integer $s \in[m]$. Let $s^*\,\deff\,\underset{s\in [m]}{\argmax} \ a(s)$ and $\tilde{s}$ denote the solution to the same optimization problem with a difference that the underlying variable $s$ is assumed to be continuous, i.e., $\tilde{s} \,\deff\, \underset{s\in \R, 0\leq s\leq m}{\argmax} \ a(s)$. One can check that $a(s)$ is a concave function which implies that $s^*$ is either equal to $\ceil{\tilde{s}}-1$ or $\ceil{\tilde{s}}$, whichever maximizes $a(s)$ and also belongs to $[m]$. 
	 The concavity of $a(s)$ also implies that $\tilde{s}$ is a root of $\diff{a(s)}{s}$ or it is equal to one of the boundary values.
	 The roots of $\diff{a(s)}{s}$ are 
	 \be{roots}
	 \tilde{s}_{\pm}=\frac{\sqrt{m(m+1)(m(1-r)+2)r}\pm(m+1)}{mr-1}.
	 \ee
	 One can check that $\tilde{s}_+$ is not a feasible solution since it does not satisfy the constraints of the continuous optimization problem, i.e., it is always the case that $\tilde{s}_+ < 0$ or $m<\tilde{s}_+$. Furthermore, we have $0<\tilde{s}_-\leq m$ for $\frac{1}{m} \leq r$, and, $m<\tilde{s}_-\leq m+1$, otherwise. Note that for the latter case the boundary condition implies $\tilde{s}=m$. Then, at least one of $\ceil{\tilde{s}}-1$ and $\ceil{\tilde{s}}$ is always a feasible solution for the discrete optimization problem and, consequently, $ s^*$ is either equal to $\ceil{\tilde{s}}-1$ or $\ceil{\tilde{s}}$, whichever is feasible and returns a larger value for $a(s)$.
	  The proof of $T$-privacy is similar to the one in LCC by noting that both the number of random mask matrices in the encoding polynomial of FLCC, characterized in \eqref{folded-Lagrange-polynomial},  and the number of shares available at each worker node are larger than those in LCC by a multiplicative factor of $m$. Hence, the dataset is perfectly masked by $mT$ random matrices $\bZ_j$'s against any coalition of $T$ worker nodes, each having $m$ evaluations of $u_m(\cdot)$.     
\end{proof} 

In order to compare the performance of FLCC with LCC, we consider evaluating $g(\cdot)$ over $m$ batches of data where each batch contains $K$ input matrices, as explained in Section \ref{Section: System Model}. To this end, LCC is run $m$ times in the first scenario, each time computing $g(\cdot)$ over a single batch of matrices. Then, the total amount of computations performed at the master node for decoding is $m$ times the decoding complexity of running LCC once. More specifically, the overall decoding complexity when LCC is employed is $O(m(N-S)\log^2(N-S)\log\log(N-S)r'h')$. Furthermore, \eqref{Lagrange-lowerbound} implies that the maximum number of Byzantine workers tolerated in this scenario can be expressed as follows 
\be{LCC-adversaries} 
A_{\rm {LCC}}\,\deff\,\floor{\frac{N-(K+T-1)D_2-S-1}{2}}.
\ee 
In FLCC, the computations performed at the master node can be considered as two separate procedures. The first one is running the list-decoding algorithm and pruning it as well as solving \eqref{extra-eqs} to uniquely interpolate $f(\cdot)$, referred to as the interpolation step. The second one corresponds to computing evaluations of $g(\cdot)$ over the set of extra evaluation points returned by the underlying pruning algorithm used and is referred to as the extra computation step. The computational complexity of the interpolation step is $O((N^2m^2+m^2K^2s)r'h')$ according to \Tref{main-result}. In the extra computation step, the master node evaluates $g(\cdot)$ over at most $s^*-1$ points. This implies that the amount of extra computations performed over the master node normalized by the computational complexity of each worker node, referred to as \emph{normalized} extra computation, is at most $\frac{s^*-1}{m}$. Moreover, the computation load of the worker nodes in both FLCC and LCC is the same, i.e., each worker node evaluates $g(\cdot)$ over a batch of data consisting of $m$ matrices in either of these scenarios. Also, the amount of commutations required, referred to as the \emph{communication complexity}, in FLCC is equal to that of running LCC $m$ times as well. The result of \Tref{FLCC-decoder-nodes} implies that the number of adversaries tolerated in FLCC is expressed as follows
\be{FLCC-adversaries}
\small{A_{\rm{FLCC}}\,\deff\,\\ \floor{\frac{s^*}{s^*+1}\bigl(N-\frac{mD_2}{m-s^*+1}(K+T-\frac{1}{m})-S-1\bigr)},}
\ee
where $s^*$ is characterized in \Tref{FLCC-decoder-nodes}. Note that by setting $m=1$, \eqref{FLCC-adversaries} is reduced to \eqref{LCC-adversaries} since an FRS code with $m=1$ is an RS code, implying that LCC and FLCC are in fact identical for this special case, as expected. Note also that the decoding complexity in LCC and interpolation complexity in FLCC grow linearly with the dimensions of the output, i.e., $r'$ and $h'$, as in both schemes the decoding procedure must be performed element-wise for all the elements of the output matrix and also both are independent of the size of input matrices, i.e., $r \times h$, as well as the degree of polynomial function evaluated over the dataset, i.e., $D_2$. However, the decoding complexity in FLCC is quadratic in the number of worker nodes $N$ while it is \emph{almost} linear in $N$ in LCC. In Figure\,\ref{costvsm}, the normalized extra computation is plotted versus the folding parameter $m$, for a certain set of parameters. It is illustrated that the ratio of the computational cost of extra evaluations of $f(\cdot)$ at the master node to the workload of a worker node approaches zero as $m$ grows. In particular, by using the result of \Tref{FLCC-decoder-nodes}, it can be observed that this ratio approaches zero in $O(\frac{1}{\sqrt{m}})$.  For instance,  when $m=100$ and $N=1000$, one can observe that the amount of extra computations needed is less than $5\%$ of the computational task of each individual worker node or, equivalently,  is less than $0.005\%$ of the total computational job.   Figure\,\ref{gainvsm} demonstrates the advantage of FLCC over LCC by comparing the maximum number of adversaries tolerated in each scheme for the same set of parameters. The ratio of the maximum number of adversaries tolerated in FLCC to that of LCC is plotted versus the folding parameter $m$. It shows that FLCC can tolerate \emph{almost} twice as many as adversaries tolerated in LCC for the same parameters $N,S,K$ and $T$. 

\begin{figure*}
    \begin{minipage}{0.45\textwidth}
        \centering
        \vspace{-3mm}\includegraphics[width=\linewidth]{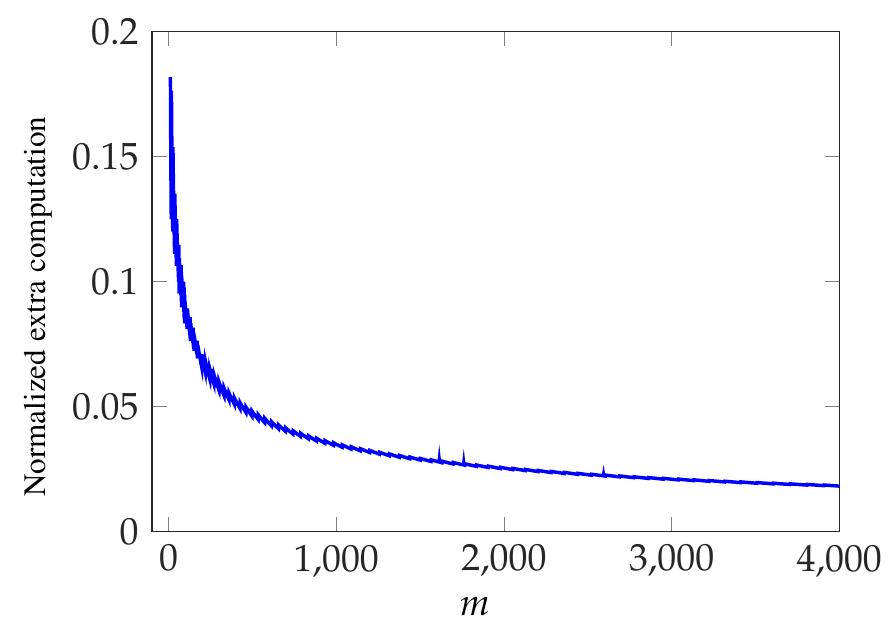}
        \caption{\small Demonstration of the ratio between the extra computations performed at the master node to the workload of each worker ($\frac{s^*-1}{m}$), referred to as \emph{normalized extra computation}, versus the folding parameter $m$ in FLCC. The relative computational cost of evaluating $f(\cdot)$ over the set of extra points at the master node approaches zero in $O(\frac{1}{\sqrt{m}})$. Other parameters are $N=1000, K=180, T=11, S=20,$ and $D=2$. }\label{costvsm}
    \end{minipage}\hfill
    \begin{minipage}{0.47\textwidth}
    \vspace{7.7mm}
        \centering
        \vspace{-19mm}\includegraphics[width=\linewidth]{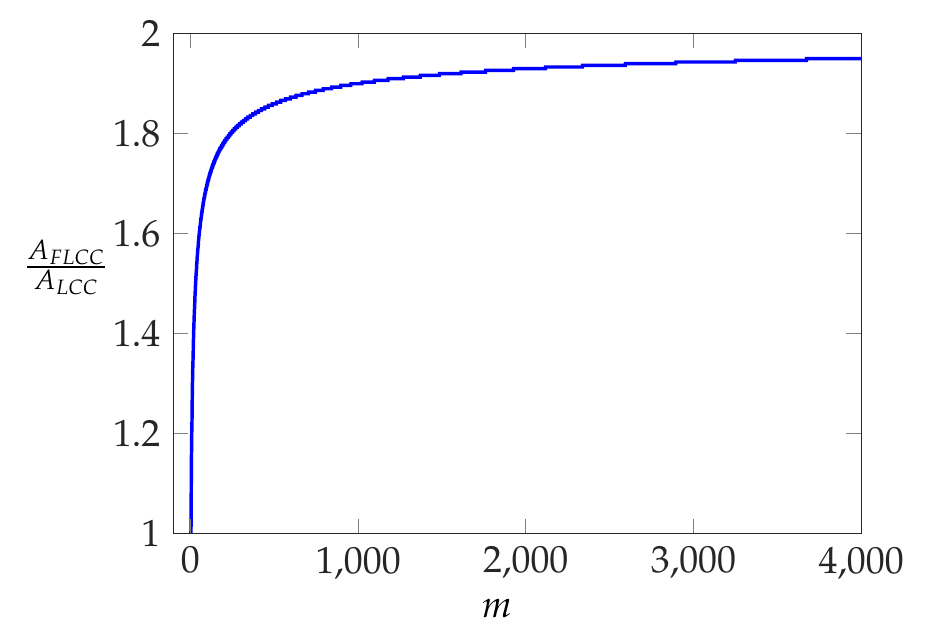}
\vspace{-5mm}\caption{\small Demonstration of the ratio of the number of Byzantine adversaries in FLCC to that of LCC for $N=1000, K=180, T=11, S=20,$ and $D=2$. The plot indicates that as $m$ grows, FLCC can tolerate \emph{almost} as twice as the number of adversaries in LCC with the same set of parameters. }\label{gainvsm}
    \end{minipage}
\end{figure*}

\noindent
The result of \Tref{FLCC-decoder-nodes} is simplified for large enough $m$ in the following corollary.
\begin{corollary}
    The FLCC specified in \Tref{FLCC-decoder-nodes} can tolerate up to 
\be{FLCC-bound-asymptotic}
A_{\rm{FLCC}}=\floor{(1-\epsilon)(N-S)-(K+T)D_2-1}
\ee
Byzantine adversaries for $m=O(\frac{1}{\epsilon^2})$ and $s^*=O(\frac{1}{\epsilon})$. 
\end{corollary}

\begin{remark}
Note that \eqref{FLCC-bound-asymptotic} implies the marginal cost of tolerating one more Byzantine adversary in FLCC is one additional worker node, the same as that of tolerating one more straggler. This demonstrates the advantage of FLCC over LCC in which two additional worker nodes are needed to tolerate one more Byzantine adversary. In other words, FLCC improves the trade-off between the number of adversaries and stragglers can be tolerated while other parameters are fixed by removing the factor $2$ in \eqref{Lagrange-lowerbound}, thereby providing a scheme in which both adversaries and stragglers cost evenly, as opposed to LCC. 
\end{remark}

\begin{remark}
In this paper, we have assumed that the extra computations needed to obtain the side information are done at the master node. The numerical results shown in Figure\,\ref{costvsm} confirms the soundness of this approach in practice. However, for scenarios where even this tiny amount of extra computation must be avoided, the master node can employ a few trusted nodes that can perform this computation without error, e.g., by using software guard  extensions (SGX) technology implemented in Intel central processing units (CPU).
\end{remark}

The decoding algorithm provided for FLCC in this section always guarantees uniquely recovering the computation outcome, i.e., the computation result is deterministically provided by the decoder. This algorithm is established upon the novel deterministic pruning algorithm for FRS code characterized in \Tref{main-result} in Section\,\ref{sec:listdecoding-results}. In this algorithm, the extra evaluation points are determined \emph{after} the list-decoding algorithm is applied. In other words, it is assumed that the side information symbols are allowed to be constructed based on the output of the list-decoding algorithm. In a practical scenario where parallelization of tasks is preferred to reduce the latency, the side information can be specified simultaneously by the master node as the workers perform computations, as shown by \Tref{probability}. In this case, the evaluation of $f(\cdot)$ over $s^*-1$ points picked uniformly at random from $\Fq$ are provided as the side information to the list decoder. \Tref{probability} implies that the system of linear equations specified in \eqref{extra-eqs} has a unique solution with \emph{high} probability, establishing that each element of the computation outcome can be uniquely determined with the same probability.  

\section{Conclusion}
\label{sec:Conclusion}
In this work, we considered a coded distributed computing setting with a master node and  a set of workers. We proposed a coding-theoretic approach that boosts the adversarial toleration threshold in such systems. In particular, we adapted the folding technique in coding theory to the context of coded computing and leveraged the list-decoding algorithms for FRS codes for recovering the overall computation outcome at the master node. Furthermore, in order to guarantee unique recovery of the outcome, we proposed novel deterministic and probabilistic pruning algorithms for list-decoding FRS codes with side information that are of independent interest in the list-decoding literature. By utilizing our proposed techniques, we introduced the folded Lagrange coded computing (FLCC) protocol that outperforms LCC by improving the number of adversaries that can be tolerated \emph{almost} by a factor of two. More specifically, we showed that in FLCC adversaries and stragglers cost almost evenly in terms of the number of workers required, compared to LCC in which tolerating one adversary costs twice as overcoming one straggler.

\section*{Acknowledgment}

The authors would like to thank the anonymous reviewer for providing a simple and short proof for Theorem 1.

\bibliographystyle{IEEEtran}
\bibliography{ref}

\appendix

\textit{Proof of \Lref{lem_appendix}:}
     The first step in the algebraic list decoding algorithm when $S$ symbols are erased is to interpolate the following multivariate polynomial by using non-erased symbols:
\be{mult_pol}
Q(X,Y_1, \cdots, Y_s)=A_0(X)+A_1(X)Y_1+ \cdots +A_s(X)Y_s,
\ee
where $\deg(A_i)\leq D$ for $i=1,\cdots,s$, and  $\deg(A_0)\leq D+k-1$, with degree parameter  $D=\floor{\frac{(N-S)(m-s+1)-k+1}{s+1}}$.
Then, the number of unknown coefficients in $Q$ is $(D+1)s+D+k=(D+1)(s+1)+k-1>(N-S)(m-s+1)$. The interpolation requirements are
\be{interpolation_req}
Q(\gamma^{im+j}, y_{im+j}, \cdots, y_{im+j+s-1})=0
\ee
for all indices $i$ corresponding to the non-erased symbols and $j=0,1,\cdots, m-s$. Since $S$ symbols are erased, the number of interpolation conditions is $(N-S)(m-s+1)$ which is less than the number of unknown coefficients in $Q$. Hence, a solution for $Q(\cdot)$ can be found by solving a homogeneous linear system over $\Fq$ with the same complexity claimed in \Lref{lin-alg-list-decoding}. By using the interpolation requirements, one can show that $Q(X, f(X), f(\gamma X), \cdots, f(\gamma^{s-1}X))=0$ if $f(\cdot)$ is a polynomial of degree at most $k-1$ whose FRS encoding agrees with the received word in  $t>\frac{D+k-1}{m-s+1}$ locations. For our choice of $D$, the requirement on $t$ is met if $t(m-s+1)>\frac{(N-S)(m-s+1)+s(k-1)}{s+1}$. Alternatively,  the fractional disagreement is at most $\frac{s}{s+1}(1-\frac{mk}{(N-S)(m-s+1)})$.

\end{document}